\documentclass[12pt]{article}
\usepackage{amssymb,amsmath,amsfonts,amsthm,lscape,xcolor,enumerate}
\usepackage[color,all]{xy}

  \usepackage{paralist}
  \usepackage{graphics} 
  \usepackage{epsfig} 
  \usepackage{graphicx}  \usepackage{epstopdf}
 \usepackage[colorlinks=true]{hyperref}
 
\hypersetup{urlcolor=blue, citecolor=red}
\newcommand{\sbt}{\,\begin{picture}(-1,1)(-1,-3)\circle*{2}\end{picture}\ }

\usepackage[all]{xy}

\usepackage{amsfonts}
\usepackage{amssymb}
\usepackage{amsmath}
\usepackage{amsthm}

\newtheorem{theorem}{Theorem}[section]

\newtheorem{lemma}[theorem]{Lemma}
\newtheorem{proposition}{Proposition}

\theoremstyle{definition}
\newtheorem{definition}[theorem]{Definition}

\newtheorem{example}[theorem]{Example}

\begin{document}

\thispagestyle{empty}

\centerline{\Large \bf Jacobi--Lie systems: Fundamentals}

\bigskip

\centerline{\Large \bf and low-dimensional classification}

\medskip
\medskip
\medskip

\medskip
\medskip

\centerline{\scshape F.J. Herranz}
\medskip
{\footnotesize
 \centerline{Department of Physics, University of Burgos,}
   \centerline{09001, Burgos, Spain.}

} 
\medskip

\centerline{\scshape J. de Lucas}
\medskip
{\footnotesize
 \centerline{Department of Mathematical Methods in Physics, University of Warsaw,}
   \centerline{ul. Pasteura 5, 02-093, Warszawa, Poland.}

}

\medskip

{\centerline{\scshape C. Sard\'on\footnote{Based on the contribution presented at ``The 10th AIMS Conference on Dynamical Systems, Differential Equations and Applications", July 07--11, 2014, Madrid, Spain. To appear in the Proceedings of the 10th AIMS Conference.}
}}\medskip
{\footnotesize
 \centerline{Department of Fundamental Physics, University of
Salamanca,}
   \centerline{Plza. de la Merced s/n, 37.008, Salamanca, Spain.}

}

\bigskip

\begin{abstract}
\noindent
A {\it Lie system} is a system of differential equations describing the integral curves of a $t$-dependent vector field taking values in a finite-dimensional real Lie algebra of vector fields, a {\it Vessiot--Guldberg Lie algebra}. We define and analyze Lie systems possessing a Vessiot--Guldberg Lie algebra of Hamiltonian vector fields relative to a Jacobi manifold, the hereafter called {\it Jacobi--Lie systems}. We classify Jacobi--Lie systems on $\mathbb{R}$ and $\mathbb{R}^2$. Our results shall be illustrated through examples of physical and mathematical interest.
\end{abstract}

\bigskip

\noindent
{\bf MSC class}: 34A26; Secondary: 53B50.
\bigskip


\newpage

\section{Introduction}

The inspection of Lie systems traces back to Lie, who laid down the fundamentals on the study of systems of first-order ordinary differential
equations admitting {\it superposition rules}, i.e., functions describing their general solutions in terms of a finite generic family of particular solutions and some parameters \cite{LS}.

Subsequently, Lie systems were hardly ever investigated for almost a century until Winternitz retook their analysis \cite{Dissertationes,PW}. Since then, many authors have been studying them \cite{ADR11,BCHLS13Ham,CGM00,CLS12Ham,Fi14}.
Although Lie systems rarely occur in the physics and mathematics literature, they appear in relevant problems and enjoy interesting geometric properties, which motivate their study \cite{Dissertationes,CLS12Ham,CV14,LV14}.

Some attention has been paid to Lie systems admitting a Vessiot--Guldberg Lie algebra (VG Lie algebra) of Hamiltonian vector fields with respect to a geometric structure \cite{ADR11,BCHLS13Ham,CGLS,CGM00,CLS12Ham,Ru10,LV14}.
First, Lie systems with VG Lie algebras of Hamiltonian vector fields relative to
symplectic and Poisson structures were briefly studied in \cite{CGM00}. Lie systems
admitting a VG Lie algebra of Hamiltonian vector fields relative to a Poisson bivector were posteriorly dubbed as {\it Lie--Hamilton systems} and carefully analyzed in \cite{CLS12Ham}, where many of their applications can be found. Next, Lie systems with VG Lie algebras of Hamiltonian vector fields relative to Dirac structures \cite{CGLS} and $k$-symplectic structures \cite{LV14} were investigated. All these geometries allow one to obtain superposition rules, constants of motion and Lie symmetries for Lie systems by means of algebraic and geometric methods, e.g., the superposition rule for Riccati equations can be obtained via the Casimir element of $\mathfrak{sl}(2,\mathbb{R})$ \cite{BCHLS13Ham}.

Following the above research, we now study Lie systems with VG Lie algebras of Hamiltonian vector fields with respect to Jacobi manifolds. Roughly speaking, a {\it Jacobi manifold} is a manifold $N$ endowed with a local Lie algebra $(C^\infty(N),\{\cdot,\cdot\})$ \cite{Co87,Do87,Ki76,LM87,Ma91,Ry00,IV}. Since Poisson manifolds are a particular case of Jacobi manifolds, Jacobi--Lie systems are a generalization of Lie--Hamilton systems. For instance, we show that Riccati equations on $\mathbb{R}$  are not Lie--Hamilton systems but they  are   Jacobi--Lie systems. The main difference between Jacobi--Lie systems and Lie--Hamilton systems is that Jacobi manifolds do not naturally give rise to Poisson brackets on a space of smooth functions on the manifold, which makes difficult to prove analogues and/or extensions of the results for Lie--Hamilton systems.

Although each generic Jacobi manifold leads to an associated Dirac manifold, not all Hamiltonian vector fields with respect to the Jacobi manifold become Hamiltonian with respect to its associated Dirac manifold (cf.~\cite{Co87}). Hence, not every Jacobi--Lie system can straightforwardly be understood as a Dirac--Lie system. Even in that case, the Jacobi manifold allows us to construct a Dirac manifold to study the system. Indeed, Dirac--Lie systems were determined through presymplectic and Poisson structures (cf.~\cite{CGLS}). Now Jacobi structures can also accomplish this task.

We here extend to Jacobi--Lie systems some of the main structures found for Lie--Hamilton systems, e.g., Lie--Hamiltonian structures \cite{CLS12Ham}, and we classify Jacobi--Lie systems on $\mathbb{R}$ and $\mathbb{R}^2$ by determining all VG Lie algebras of Hamiltonian vector fields with respect to Jacobi manifolds on $\mathbb{R}$ and $\mathbb{R}^2$. This is achieved by using the local classification of Lie algebras of vector fields on $\mathbb{R}$ and $\mathbb{R}^2$ derived  by Lie \cite{1880} and improved by  Gonz{\'a}lez-L{\'o}pez,  Kamran and   Olver (GKO) \cite{GKP92} (see also~\cite{LH2013}). As a result, we
obtain that every Lie system on $\mathbb{R}$ is a Jacobi--Lie system and, furthermore, we also show that every Jacobi--Lie system on $\mathbb{R}^2$ admits a VG Lie algebra diffeomorphic to one of the 14 classes indicated in Table \ref{table1} below.

\section{Fundamentals on Lie--Hamilton systems}
 All structures throughout this work are   assumed to be smooth, real and globally defined.
Let $V$ be a vector space and $[\cdot,\cdot]:V\times V\rightarrow V$
a Lie bracket. We denote by $(V, [ \cdot, \cdot])$ the corresponding real Lie algebra.
 Given subsets $\mathcal{A}_1,\mathcal{A}_2
 \subset V$, we write $[ \mathcal{A}_1, \mathcal{A}_2]$ for the real
vector space spanned by the Lie brackets between elements of $\mathcal{A}_1$ and
$\mathcal{A}_2$, respectively. We define Lie$(\mathcal{A}_1, V, [ \cdot,
\cdot])$ to be the smallest Lie subalgebra of $V$ containing $\mathcal{A}_1$.
To abbreviate, we use Lie$( \mathcal{A}_1)$ and $V$ to represent
Lie$( \mathcal{A}_1, V, [ \cdot, \cdot])$ and $(V, [ \cdot, \cdot])$, correspondingly.

A {\it $t$-dependent vector field} on $N$ is a map $X : (t, x)
\in
\mathbb{R} \times N \mapsto X (t, x) \in TN$ such that $X_t:x\in N\mapsto X(t,x)\in TN$ is a vector field for
each $t\in\mathbb{R}$
{\cite{Dissertationes}}.
We call {\it minimal Lie algebra} of $X$  the
Lie algebra $V^X \equiv {\rm Lie} (\{X_t \}_{t \in \mathbb{R}})$. An {\it integral curve} of $X$ is an integral curve  $\gamma:
\mathbb{R} \mapsto \mathbb{R} \times N$ of the vector field $\partial / \partial t+X (t, x)$ on $\mathbb{R} \times
N$. For every $\gamma$ of the form $t\mapsto (t,x(t))$, we have that
\begin{equation}\label{normalform}\nonumber
\frac{{\rm d}x}{{\rm d} t} ( t) = (X \circ \gamma) (
   t) .
\end{equation}
This system is referred to as the {\it associated system} of $X$.
Conversely, every system of first-order differential equations in the normal form
describes the integral curves $(t,x(t))$ of a unique $t$-dependent vector field $X$.
So, we can use $X$ to denote both the $t$-dependent vector field and the associated system.

\begin{definition}
A {\it Lie system} is a system $X$ whose $V^X$ is finite-dimensional \cite{Dissertationes}.
\end{definition}
Note that if $X$ admits a VG Lie algebra, then $V^X$ is finite-dimensional and conversely.
\begin{example}
Consider the following system of Riccati equations \cite{Dissertationes}
\begin{equation}\label{coupRiceq}
\frac{{\rm d}x_i}{{\rm d}t}=a_0(t)+a_1(t)x_i+a_2(t)x_i^2,\qquad i=1,\ldots,n,
\end{equation}
with $a_0(t),a_1(t),a_2(t)$ being arbitrary $t$-dependent functions.
System \eqref{coupRiceq} is associated to the $t$-dependent vector field $X_{R}=a_0(t)X_1+a_1(t)X_2+a_2(t)X_3,$
where
\begin{equation}\label{VF}
X_1=\sum_{i=1}^n\partial_{x_i},\qquad X_2=\sum_{i=1}^nx_i\partial_{x_i},\qquad X_3=\sum_{i=1}^nx_i^2{\partial_{x_i}}.
\end{equation}
Hence, $X_{R}$ takes values in the VG Lie algebra $\langle X_1,X_2,X_3\rangle \simeq \mathfrak{sl}(2,\mathbb{R})$ and becomes a Lie system. Lie proved that each Lie system on $\mathbb{R}$ is locally diffeomorphic to a particular case of (\ref{coupRiceq}) for $n=1$ \cite{Dissertationes,1880}.
\end{example}

Let $\Gamma(\Lambda^2 TN)$ be the space of sections of $\Lambda^2TN$. A {\it Poisson manifold} is a pair $\left( N, \Lambda \right)$,
with $\Lambda$ being a bivector field on $N$, i.e., $\Lambda\in \Gamma(\Lambda^2 TN)$, satisfying that $[\Lambda,\Lambda]_{SN}=0$, where $[\cdot,\cdot]_{SN}$ is the {\it Schouten--Nijenhius bracket}\footnote{\tiny{$[X_1\wedge\ldots\wedge X_p,Y_1\wedge\ldots\wedge Y_q]_{SN}\equiv(-1)^{p+1}\sum^{1\leq i\leq p}_{1\leq j\leq q}(-1)^{i+j}[X_i,Y_j]\wedge X_1\wedge\stackrel{\widehat X_i}{\ldots} \wedge X_p\wedge Y_1\wedge\stackrel{\widehat Y_j}{\ldots}\wedge Y_q$.}}\cite[p.7]{IV}. The bivector $\Lambda$, the so-called {\it Poisson bivector}, induces a bundle morphism $\widehat \Lambda:\alpha_x \in T^*N\mapsto \Lambda_x(\alpha_x,\cdot)\in TN$.
We say that a vector field $X$ on $N$ is   Hamiltonian with respect to $(N,\Lambda)$ if 
 $X=\widehat \Lambda ({\rm d}f)$ for an $f\in C^\infty(N)$. We call $f$ a {\it Hamiltonian function} for $X$.
Conversely, every function $f$ is the Hamiltonian function of a unique vector field $X_f$, its {\it Hamiltonian vector field}.
This gives rise to the bracket
 $\{\cdot,\cdot \}_\Lambda:(f,g)\in C^{\infty} \left( N \right)\times C^\infty(N)\mapsto \{f,g\}_\Lambda\equiv\Lambda({\rm d}f,{\rm d}g)\in C^\infty(N)$. This bracket, the {\it Poisson bracket},  turns $C^\infty(N)$ into a {\it Poisson algebra} $(C^\infty(N),\sbt\,,\{\cdot,\cdot\}_\Lambda)$, i.e.,
$\{\cdot, \cdot \}_\Lambda$ is a Lie bracket on $C^\infty(N)$ and $\{f,g \sbt\, h\}_\Lambda = \{ f, g\}_\Lambda \sbt\ h + g \sbt\ \{ f, h\}_\Lambda$, $\forall f, g, h \in C^\infty(N).$ Every Poisson bracket on $C^\infty(N)$ amounts to a Poisson bivector on $N$ \cite{IV}.

\begin{definition}
A {\it Lie--Hamilton system} is a Lie system $X$ whose
$V^X$ consists of Hamiltonian vector fields with respect to a Poisson bivector \cite{CLS12Ham}.
\end{definition}

\begin{example}
Let us reconsider the Lie system (\ref{coupRiceq}) for $n=4$ and
\begin{equation*}
\Lambda_R=(x_1-x_2)^2{\partial_{x_1}}\wedge \partial_{x_2}+(x_3-x_4)^2{\partial_{ x_3}}\wedge {\partial_{ x_4}},
\end{equation*}
which satisfies $[\Lambda_R,\Lambda_R]_{SN}=0$ on  $\mathcal{O}=\{(x_1,x_2,x_3,x_4)|(x_1-x_2)(x_3-x_4)\neq 0\}\subset \mathbb{R}^4$.
We have that $X_i=\widehat\Lambda_R({\rm d}h_i)$ on $(\mathcal{O},\Lambda_R)$ for $X_1,X_2,X_3$ given in \eqref{VF} and
{\begin{footnotesize}
\[
\begin{gathered}
{h_1}=-\frac{1}{x_1-x_2}-\frac{1}{x_3-x_4},\quad
h_2=-\frac 12\left(\frac{x_1+x_2}{x_1-x_2}+\frac{x_3+x_4}{x_3-x_4}\right),\quad
h_3=-\frac{x_1 x_2}{x_1-x_2}-\frac{x_3 x_4}{x_3-x_4}.
\end{gathered}
\]
\end{footnotesize}}Hence, $X_1,X_2,X_3$ are Hamiltonian relative to $(\mathcal{O},\Lambda_{R})$ and $X_{R}$ is a Lie--Hamilton system. This fact can be used to derive the superposition rule for Riccati equations through a Casimir of $\mathfrak{sl}(2,\mathbb{R})$ \cite{BCHLS13Ham}.
\end{example}
\section{Jacobi manifolds}

Jacobi manifolds were introduced by Kirillov and Lichnerowicz \cite{Ki76,Li77}. We now briefly survey their most fundamental properties \cite{Ma91,IV}.

\begin{definition}A {\it Jacobi manifold} is a triple $(N,\Lambda,R)$, where $\Lambda$ is a bivector field on $N$ and $R$ is a vector field on $N$, the referred to as {\it Reeb vector field}, satisfying
$$
[\Lambda,\Lambda]_{SN}=2R\wedge \Lambda,\qquad [R,\Lambda]_{SN}=0.
$$
\end{definition}

\begin{example} Every Poisson manifold $(N,\Lambda)$ can be considered as a Jacobi manifold $(N,\Lambda,R=0)$.
\end{example}

\begin{example}\label{BivectorRH}
The {\it continuous Heisenberg group} \cite{We00} can
be described as the space of matrices
\begin{equation}\label{Hei}
\mathbb{H}=\left\{\left(\begin{array}{ccc}1&x&z\\0&1&y\\0&0&1\end{array}
\right)\bigg|\,x,y,z\in\mathbb{R}\right\},
\end{equation}
endowed with the standard matrix multiplication,
where $\{x,y,z\}$ is the natural coordinate system on $\mathbb{H}$ induced by (\ref{Hei}). Consider the bivector field on $\mathbb{H}$ given by
\begin{equation}\label{BivectorH}
\Lambda_\mathbb{H}\equiv -y\partial_y\wedge\partial_ z+\partial_x\wedge\partial_y
\end{equation}
and the vector field $R_\mathbb{H}\equiv \partial_z$. After a simple calculation, we obtain that
$$
[\Lambda_\mathbb{H},\Lambda_\mathbb{H}]_{SN}=2\partial_ x\wedge \partial_ y\wedge\partial_ z= 2R_\mathbb{H}\wedge \Lambda_\mathbb{H},\qquad [R_\mathbb{H},\Lambda_\mathbb{H}]_{SN}=0.
$$
So, $(\mathbb{H},\Lambda_\mathbb{H},R_\mathbb{H})$ is a Jacobi manifold.
\end{example}
\begin{definition} We say that $X$ is a {\it Hamiltonian
vector field} with respect to the Jacobi manifold $(N,\Lambda,R)$ if there exists a function $f\in C^\infty(N)$ such
that
$$
X=[\Lambda,f]_{SN}+fR=\widehat\Lambda({\rm d}f)+fR.
$$
We say that $f$ is a {\it Hamiltonian function} of $X$ and we write $X=X_f$. If $Rf=0$, we call $f$ 
a {\it good Hamiltonian function} and $X_f$ a {\it good Hamiltonian vector field}.
\end{definition}
\begin{example} Given the Jacobi manifold $(\mathbb{H},\Lambda_\mathbb{H},R_\mathbb{H})$ and $X_1^L\equiv \partial_ x$, we have
$$
X_1^L=[\Lambda_\mathbb{H},-y]_{SN}-yR_\mathbb{H}=\widehat{\Lambda}_\mathbb{H}(-{\rm d}y)-yR_\mathbb{H}.
$$
Hence, $h_1^{L}=-y$ is a Hamiltonian function for $X_1^L$ with respect to $(\mathbb{H},\Lambda_\mathbb{H},R_\mathbb{H})$.
\end{example}
Each function gives rise to a unique Hamiltonian vector field. Nevertheless, a vector field may admit several Hamiltonian functions. This fact will be illustrated in Lemma \ref{63}. We write ${\rm Ham}(N,\Lambda,R)$ for the space of Hamiltonian vector fields of $(N,\Lambda,R)$. It is well known that ${\rm Ham}(N,\Lambda,R)$ is a Lie algebra relative to the Lie bracket of vector fields, and $(C^\infty(N),\{\cdot,\cdot\}_{\Lambda,R})$ becomes a local Lie algebra with
$$
\{f,g\}_{\Lambda,R}=\Lambda({\rm d}f,{\rm d}g)+f Rg-gRf.
$$
This Lie bracket becomes a Poisson bracket if and only if $R=0$. Moreover, the morphism $\phi_{\Lambda,R}:f\in C^\infty(N)\mapsto X_f\in {\rm Ham}(N,\Lambda,R)$ is a Lie algebra morphism.  

\section{Jacobi--Lie systems}
We now
introduce Jacobi--Lie systems as Lie systems admitting a VG Lie
algebra of Hamiltonian vector fields relative to a Jacobi manifold.

\begin{definition} A {\it Jacobi--Lie system} $(N,\Lambda,R,X)$ consists of a
Jacobi manifold $(N,\Lambda,R)$ and a Lie system $X$  satisfying that $V^X\subset {\rm Ham}(N,\Lambda,R)$.
\end{definition}

\begin{example}
A straightforward calculation shows that the Lie
algebra, $\mathfrak{h}$, of left-invariant vector fields on $\mathbb{H}$ is
spanned by
$$
X^L_1=\partial_x,\qquad X^L_2=\partial_
y+x\partial_z,\qquad X^L_3=\partial_ z.
$$
Consider now the system on $\mathbb{H}$ given by
\begin{equation}\label{H}
\frac{{\rm d}{\mathcal H}}{{\rm d}t}=\sum_{i=1}^3b_i(t)X^L_i(\mathcal H),\qquad \mathcal H\in\mathbb{H},
\end{equation}
for arbitrary $t$-dependent functions $b_i(t)$. Since the associated $t$-dependent vector field $X^\mathbb{H}=\sum_{i=1}^3b_i(t)X_i^L$ takes values in
a finite-dimensional Lie algebra of vector fields, then $X^\mathbb{H}$ is a Lie system.
The interest of $X^\mathbb{H}$ is due to its appearance in the solution of
the so-called {\it quantum Lie systems} as well as Lie systems admitting a VG Lie algebra
isomorphic to $\mathfrak{h}$ (cf. ~\cite{Dissertationes}).

Let us show that the Lie system (\ref{H}) leads to a Jacobi--Lie system $(\mathbb{H},\Lambda_\mathbb{H},R_\mathbb{H},X^\mathbb{H})$, where $\Lambda_{\mathbb H}$ and $R_\mathbb{H}$ are those appearing in Example \ref{BivectorRH}. Note that,
$$
X^L_1=[\Lambda_{\mathbb H},-y]_{SN}-yR_\mathbb{H},\quad X^L_2=[\Lambda_{\mathbb H},x]_{SN}+xR_\mathbb{H},\quad X^L_3=[\Lambda_{\mathbb H},1]_{SN}+R_\mathbb{H}.
$$
That is, $X^L_1,X^L_2$ and $X^L_3$ are Hamiltonian vector fields with Hamiltonian functions $h_1=-y$, $h_2=x$ and $h_3=1$, respectively.
Hence, $(\mathbb{H},\Lambda_\mathbb{H},R_\mathbb{H},X^\mathbb{H})$ is a Jacobi--Lie system. Note that each $h_i$ is a first-integral of $X_i^L$ and $R_\mathbb{H}$ for $i=1,2,3$, respectively.
\end{example}
\begin{example}
Consider the Lie group $\mathbb{G}\equiv SL(2,\mathbb{R})$ of matrices $2\times 2$ with real entries $\alpha,\beta,\gamma,\delta$ satisfying $\alpha\delta-\beta\gamma=1$. Close to its neutral element, $\{\alpha,\beta,\gamma\}$ form a local coordinate system for $\mathbb{G}$. A short calculation shows that 
$$
X^R_1=\alpha\partial_\alpha+\beta \partial_\beta-\gamma\partial_\gamma,\qquad X^R_2=\gamma\partial_\alpha+\frac{1+\beta\gamma}{\alpha}\partial_ \beta,\qquad X^R_3=\alpha\partial_ \gamma
$$
span the space of right-invariant vector fields on $\mathbb{G}$. If we define
\begin{equation}\label{JMG}
\Lambda_\mathbb{G}=\alpha\beta\partial_\alpha\wedge\partial_\beta-(1+\beta\gamma)\partial_\beta\wedge\partial_\gamma,\qquad R_\mathbb{G}=\alpha\partial_\alpha-\beta\partial_\beta+\gamma\partial_\gamma,
\end{equation}
we obtain that $[\Lambda_{\mathbb{G}},\Lambda_{\mathbb{G}}]_{SN}=-2\alpha\partial_\alpha\wedge\partial_\beta\wedge\partial_\gamma=2R_\mathbb{G}\wedge \Lambda_{\mathbb{G}}$ and $[R_{\mathbb{G}},\Lambda_{\mathbb{G}}]_{SN}=0$. So, $(\mathbb{G},\Lambda_\mathbb{G},R_\mathbb{G})$ is  a Jacobi manifold.
Consider now the system on $\mathbb{G}$ given by
$
\frac{{\rm d}\mathcal G}{{\rm d}t}=\sum_{i=1}^3b_i(t)X^R_i(\mathcal G)$, $\mathcal G\in\mathbb{G},
$ 
for any $t$-dependent functions $b_i(t)$. Since $X^\mathbb{G}=\sum_{i=1}^3b_i(t)X_i^R$ takes values in
the Lie algebra $V^\mathbb{G}=\langle X^R_1,X^R_2,X^R_3\rangle$, the system $X^\mathbb{G}$ is a Lie system. System $X^\mathbb{G}$ occurs in the study of  Briosche--Darboux--Halphen equations, Kummer--Schwarz equations, Milne--Pinney equations, etc.~\cite{CGLS,Dissertationes,EHLS15}.

We now prove that $(\mathbb{G},\Lambda_\mathbb{G},R_\mathbb{G},X^\mathbb{G})$ is a Jacobi--Lie system. In fact, $X_1^R,X^R_2,X^R_3$ are Hamiltonian  relative to $(\mathbb{G},\Lambda_\mathbb{G},R_\mathbb{G})$ with good Hamiltonian functions
\begin{equation}\label{HamG}
h_1=1+2\beta\gamma,\qquad h_2=\frac{\gamma}{\alpha}(1+\beta\gamma),\qquad h_3=-\beta\alpha.
\end{equation}
Note that these functions are first-integrals of $X_1^R,X^R_2,X^R_3$, respectively, and $R_\mathbb{G}$. This allows us to use  
$X_i^R+{\rm d}h_i$ with $i=1,2,3$, and $R_\mathbb{G}$ to span a sub-bundle $L_\mathbb{G}$ of ${T}\mathbb{G}\oplus_\mathbb{G}{T}^*\mathbb{G}$ originating a Dirac structure on $\mathbb{G}$ \cite{Co87}. Vector fields $X^R_1,X_2^R,X_3^R$ are Hamiltonian relative to $L_\mathbb{G}$ giving rise to a Dirac--Lie system $(\mathbb{G},L_\mathbb{G},X^\mathbb{G})$ \cite{CGLS}. 
\end{example}
\begin{lemma}
The space $G(N,\Lambda,R)$ of good Hamiltonian functions of a Jacobi manifold $(N,\Lambda,R)$ is a Poisson algebra relative to the bracket $\{\cdot,\cdot\}_{\Lambda,R}$, and $\star_g:f\in C^\infty(N)\mapsto \{g,f\}_{\Lambda,R}\in C^\infty(N)$,  $\forall g\in G(N,\Lambda,R)$, is a derivation on $(C^\infty(N),\sbt)$.
 \end{lemma}
\begin{proof}
First, we prove that the Jacobi bracket of two good Hamiltonian functions is a good Hamiltonian function.
For general functions $u_1,u_2\in C^\infty(N)$, we have
 $$R\{u_1,u_2\}_{\Lambda,R}=R(\Lambda({\rm d} u_1,{\rm d}u_2)+u_1Ru_2-u_2Ru_1).$$
If $u_1,u_2 \in G(N,\Lambda,R)$, then $Ru_1=Ru_2=0$. From this and $[\Lambda,R]_{SN}=0$, we have
\begin{equation*}
R\{u_1,u_2\}_{\Lambda,R}=R(\Lambda({\rm d}u_1,{\rm d}u_2))=[R,[[\Lambda,u_1]_{SN},u_2]_{SN}]_{SN}=0.
\end{equation*}
Hence, $\{u_1,u_2\}_{\Lambda,R}\in G(N,\Lambda,R)$, which becomes a Lie algebra relative to $\{\cdot,\cdot\}_{\Lambda,R}$. Note also that $R(u_1\sbt u_2)=0$ and $u_1\sbt u_2\in G(N,\Lambda,R)$.

Given an arbitrary $u_1\in G(N,\Lambda,R)$ and any $u_2,u_3\in C^\infty(N)$, we have that
\begin{multline*}
\star_{u_1}(u_2\sbt u_3)=\Lambda({\rm d}u_1,{\rm d}(u_2\sbt u_3))+u_1R(u_2\sbt u_3)-u_2\sbt u_3 Ru_1\\=X_{u_1}(u_2\sbt u_3)=u_3\sbt X_{u_1}u_2+u_2\sbt X_{u_1}u_3=u_3\star_{u_1}u_2+u_2\star_{u_1}u_3.
\end{multline*}
So, $\star_{u_1}$ is a derivation on $(C^\infty(N),\sbt)$ and also  on $(G(N,\Lambda,R),\{\cdot,\cdot\}_{\Lambda,R})$. From this it  follows that
 $(G(N,\Lambda,R),\sbt,\{\cdot,\cdot\}_{\Lambda,R})$ is a Poisson algebra.
\end{proof}

\section{Jacobi--Lie Hamiltonian systems}

\begin{definition} We call {\it Jacobi--Lie Hamiltonian system} a quadruple $(N,\Lambda,R,h)$,
where $(N,\Lambda,R)$ is a Jacobi manifold and $h:(t,x)\in\mathbb{R}\times N\mapsto h_t(x)\in
N$ is a $t$-dependent function such that
${\rm Lie}(\{h_t\}_{t\in\mathbb{R}},\{\cdot,\cdot\}_{\Lambda,R})$ is finite-dimensional. Given a system $X$ on $N$, we say that $X$ admits a {\it Jacobi--Lie Hamiltonian system}  $(N,\Lambda,R,h)$ if
$X_t$ is a Hamiltonian vector field with Hamiltonian function $h_t$ (with respect to $(N,\Lambda,R)$) for each $t\in \mathbb{R}$.
\end{definition}

\begin{example}
Observe that $h_t=\sum_{i=1}^3b_i(t)h_i= -b_1(t) y+b_2(t)x+b_3(t) $ is a Hamiltonian
function of $X^\mathbb{H}_t$ in \eqref{H} for every $t\in\mathbb{R}$. In addition,
$$
\{h_1,h_2\}_{\Lambda_\mathbb{H},R_\mathbb{H}}=h_3,\qquad \{h_1,h_3\}_{\Lambda_\mathbb{H},R_\mathbb{H}}=0,\qquad \{h_2,h_3\}_{\Lambda_\mathbb{H},R_\mathbb{H}}=0.
$$
So, the functions $\{h_t\}_{t\in\mathbb{R}}$ span a finite-dimensional real Lie
algebra of functions relative to the Lie bracket induced by (\ref{BivectorH}) and $R_\mathbb{H}=\partial/\partial z$. Thus, $X^\mathbb{H}$ admits a Jacobi--Lie Hamiltonian system $(N,\Lambda_\mathbb{H},R_\mathbb{H},h)$.
\end{example}

\begin{example}
Relative to the Lie bracket induced by $(G,\Lambda_\mathbb{G},R_\mathbb{G})$  given in (\ref{JMG}), the functions (\ref{HamG}) satisfy that 
$$
\{h_1,h_2\}_{\Lambda_\mathbb{G},R_\mathbb{G}}=-2h_2,\qquad \{h_1,h_3\}_{\Lambda_\mathbb{G},R_\mathbb{G}}=2h_3,\qquad \{h_2,h_3\}_{\Lambda_\mathbb{G},R_\mathbb{G}}=-h_1.
$$
So, $(\mathbb{G},\Lambda_{\mathbb{G}},R_{\mathbb{G}},h\equiv \sum_{i=1}^3b_i(t)h_i)$ is a Jacobi--Lie Hamiltonian system for $X^\mathbb{G}$.
\end{example}

\begin{theorem} If $(N,\Lambda,R,h)$ is a Jacobi--Lie Hamiltonian system, then the system $X$ of the form $X_t=X_{h_t}$,  $\forall t\in\mathbb{R}$, gives rise to a Jacobi--Lie
system $(N,\Lambda,R,X)$. If $X$ is a Lie system and the $\{X_t\}_{t\in\mathbb{R}}$ are good Hamiltonian vector fields, then $X$ admits a Jacobi--Lie Hamiltonian.
\end{theorem}

\begin{proof}

Let us prove the direct part. By assumption, the Hamiltonian functions $\{h_t\}_{t\in\mathbb{R}}$ are contained in a finite-dimensional Lie algebra ${\rm Lie}(\{h_t\}_{t\in\mathbb{R}}, \{\cdot,\cdot\}_{\Lambda,R})$. The Lie algebra morphism
$\phi_{\Lambda,R}:f\in C^{\infty}(N)\mapsto X_f\in $ Ham$(N,\Lambda,R)$ maps the curve $h_t$ into a curve $X_t$ within $\phi_{\Lambda,R}({\rm Lie}(\{h_t\}_{t\in\mathbb{R}}, \{\cdot,\cdot\}_{\Lambda,R})).$ Since Lie($\{h_t\}_{t},\{\cdot,\cdot\}_{\Lambda,R}$) is finite-dimensional and $\phi_{\Lambda,R}$ is a Lie algebra morphism, $\phi_{\Lambda,R}( {\rm Lie}(\{h_t\}_{t},\{\cdot,\cdot\}_{\Lambda,R}))$ is a finite-dimensional Lie algebra. Since $X$ takes values
in the latter Lie algebra of Hamiltonian vector fields, then $(N,\Lambda,R,X)$ is a Jacobi--Lie system.

 Let us prove the converse. Since the elements of $\{X_t\}_{t\in\mathbb{R}}$ are good Hamiltonian vector fields by assumption and Lie$(\{X_t\}_{t\in\mathbb{R}})=V^X$, every element of $V^X$ is a good Hamiltonian vector field and we can choose
a basis $X_1,\ldots, X_r$ of $V^X$ with good Hamiltonian functions $h_1,\ldots, h_r$. The Jacobi bracket $\{h_i,h_j\}_{\Lambda,R}$ is a good Hamiltonian function for $[X_i,X_j]$.

Since $[X_i,X_j]=\sum_{k=1}^rc_{ijk}X_k$ for certain constants $c_{ijk}$, we obtain that each $$
s_{ij}=\{h_i,h_j\}_{\Lambda,R}-\sum_{k=1}^rc_{ijk}h_k,\qquad i<j,
$$
is the difference of two good Hamiltonian functions with the same Hamiltonian vector field. Hence, $\{s_{ij},h\}_{\Lambda,R}=0$ for all $h\in C^\infty(N)$. Using this, we obtain that the linear space generated by $h_1,\ldots,h_r,s_{ij}$, with $1\leq i<j\leq r$, is a finite-dimensional Lie algebra relative to $\{\cdot,\cdot\}_{\Lambda,R}$. If $X=\sum_{i=1}^rb_i(t)X_i$, then $(N,\Lambda,R,h=\sum_{i=1}^rb_i(t)h_i)$ is a Jacobi--Lie Hamiltonian system for $X$.
\end{proof}

The following proposition can be straightforwardly proved.

\begin{proposition} Let $(N,\Lambda,R,X)$ be a Jacobi--Lie system admitting a Jacobi--Lie Hamiltonian $(N,\Lambda,R,h)$ of good Hamiltonian functions $\{h_t\}_{t\in\mathbb{R}}$. Then, $f\in C^\infty(N)$ is a $t$-independent constant of motion for $X$ if and only if $f$ commutes with all the elements of ${\rm Lie}(\{h_t\}_{t\in\mathbb{R}},\{\cdot,\cdot\}_{\Lambda,R})$ relative to $\{\cdot,\cdot\}_{\Lambda,R}$.
\end{proposition}

\begin{example} Consider again the functions $h_1,h_2,h_3$ given in (\ref{HamG}) and the Jacobi manifold $(\mathbb{G},\Lambda_\mathbb{G},R_\mathbb{G})$, with $\Lambda_\mathbb{G}$ and $R_\mathbb{G}$ given by (\ref{JMG}). Then, $\{h^2_1+4h_2h_3,h_i\}_{\Lambda_\mathbb{G},R_\mathbb{G}}=0$ for $i=1,2,3$. So, $C=h^2_1+4h_2h_3$ is a constant of motion for $X^\mathbb{G}$.
\end{example}
\section{Jacobi--Lie systems on low dimensional\\ manifolds}
We now prove that every Lie system on the real line gives rise to a Jacobi--Lie system. Next, we classify, via the GKO classification 
(see \cite{LH2013,GKP92} and Table \ref{table1}), all VG Lie algebras on $\mathbb{R}^2$ related to Jacobi--Lie systems.

\begin{table}[t] {\footnotesize
 \noindent
\caption{{\small
{\footnotesize VG LIE ALGEBRAS OF HAMILTONIAN VECTOR FIELDS ON ${\mathbb R}^2$  RELATIVE TO A JACOBI MANIFOLD.  P means Poisson. Functions $1,\xi_1(x),\ldots,\xi_r(x)$ are linearly independent and $\eta_1(x),\ldots,\eta_r(x)$ form a basis of solutions for certain $d^rf/dx^r=\sum_{\alpha=0}^{r-1}c_\alpha d^\alpha f/dx^\alpha, c_\alpha\in\mathbb{R}$.
We write $\mathfrak{g}_1\ltimes \mathfrak{g}_2$ for the (Lie algebra) semi-direct sum of $\mathfrak{g}_1$ by $\mathfrak{g}_2$, i.e., $\mathfrak{g}_2$ is an ideal of $\mathfrak{g}_1\ltimes\mathfrak{g}_2$.}}}
\label{table1}
\noindent\hfill
 \begin{tabular}{ llll}
&  &\\[-1.0ex]
\hline
&  &\\[-1.9ex]
\#&Lie algebra & Basis of vector fields $X_i$ & Jacobi
\\[+0.25ex]
\hline
 &  &\\[-1.9ex]
P$_1$&$A_\alpha\simeq \mathbb{R}\ltimes \mathbb{R}^2$ & $  { {\partial_x} ,    {\partial_y}  ,   \alpha(x\partial_x + y\partial_y)  +  y\partial_x - x\partial_y},\quad \ \alpha\geq 0$&$(\alpha=0)$ P
\\[+0.25ex]
P$_2$&$\mathfrak{sl}(2)$ & $ {\partial_x},   {x\partial_x  +  y\partial_y}  ,   (x^2  -  y^2)\partial_x  +  2xy\partial_y$&P
\\[+0.25ex]
P$_3$&$\mathfrak{so}(3)$ &${     { y\partial_x  -  x\partial _y},     { (1  +  x^2  -  y^2)\partial_x  +  2xy\partial_y}  ,    }$&
\\[+0.25ex]
 &  &${      2xy\partial_x  +  (1  +  y^2  -  x^2)\partial_y}$&P
\\[+0.25ex]
P$_4$&$\mathbb{R}^2\ltimes\mathbb{R}^2$ &$   {\partial_x},   {\partial_y} ,  x\partial_x + y\partial_y,   y\partial_x - x\partial_y$&No
\\[+0.25ex]
P$_5$&$\mathfrak{sl}(2 )\ltimes\mathbb{R}^2$ &${   {\partial_x},   {\partial_y} ,  x\partial_x - y\partial_y,  y\partial_x,  x\partial_y}$&P\\[+0.25ex]
P$_6$&$\mathfrak{gl}(2 )\ltimes\mathbb{R}^2$ &${  {\partial_x},    {\partial_y} ,   x\partial_x,   y\partial_x,   x\partial_y,   y\partial_y}$& No\\[+0.25ex]
P$_7$&$\mathfrak{so}(3,1)$ &${   {\partial_x},   {\partial_y} ,   x\partial_x\!+\! y\partial_y,   y\partial_x \!-\! x\partial_y,   (x^2 \!-\! y^2)\partial_x \!+\! 2xy\partial_y ,}$  &
\\[+0.25ex]
 &  &${    2xy\partial_x \!+\! (y^2\!-\!x^2)\partial_y}$  &No
\\[+0.25ex]
P$_8$&$\mathfrak{sl}(3 )$ &${   {\partial_x},    {\partial_y} ,   x\partial_x,   y\partial_x,   x\partial_y,   y\partial_y,   x^2\partial_x + xy\partial_y,   xy\partial_x  +  y^2\partial_y}$& No
\\[+1.5ex]
I$_1$&$\mathbb{R}$ &$   {\partial_x}  $ & P, $(0,\partial_x)$
\\[+0.25ex]
I$_2$&$\mathfrak{h}_2$ & $   {\partial_x}  ,  x\partial_x$& P, $(0,\partial_x)$
\\[+0.25ex]
I$_3$&$\mathfrak{sl}(2 )$ (type I) &$   {\partial_x} ,  x\partial_x,  x^2\partial_x$&  P, $(0,\partial_x)$
\\[+0.25ex]
I$_4$&$\mathfrak{sl}(2 )$ (type II) & ${   {\partial_x  +  \partial_y},    {x\partial _x + y\partial_y} ,   x^2\partial_x  +  y^2\partial_y}$ & P\\[+0.25ex]
I$_5$&$\mathfrak{sl}(2 )$ (type III) &${  {\partial_x},    {2x\partial_x + y\partial_y} ,   x ^2\partial_x  +  xy\partial_y}$& P
\\[+0.25ex]
I$_6$&$\mathfrak{gl}(2 )$ (type I)& ${  {\partial_x},    {\partial_y} ,   x\partial_x,   x^2\partial_x}$&No
\\[+0.25ex]
I$_7$&$\mathfrak{gl}(2 )$ (type II)& ${   {\partial_x},   {y\partial_y}  ,     x\partial_x,    x^2\partial_x +  xy \partial_y}$&No
 \\[+0.25ex]
I$_8$&$B_\alpha\simeq \mathbb{R}\ltimes\mathbb{R}^2$ &${   {\partial_x},    {\partial_y} ,   x\partial_x  +  \alpha y\partial_y},\quad  0<|\alpha|\leq 1$&$(\alpha=-1)$  P\\[+0.25ex]
I$_9$&$\mathfrak{h}_2\oplus\mathfrak{h}_2$ &${  {\partial_x},    {\partial_y} ,   x\partial_x,  y\partial_y}$&No
\\[+0.25ex]
I$_{10}$&$\mathfrak{sl}(2 )\oplus \mathfrak{h}_2$ & ${  {\partial_x},    {\partial_y}  ,   x\partial_x,  y\partial_y,  x^2\partial_x }$&No
\\[+0.25ex]
I$_{11}$&$\mathfrak{sl}(2 )\oplus\mathfrak{sl}(2 )$ &$   {\partial_x},    {\partial_y} ,   x\partial_x,   y\partial_y,   x^2\partial_x ,   y^2\partial_y $&No\\[+0.25ex]
I$_{12}$&$\mathbb{R}^{r + 1}$ &$  {\partial_y}  ,   \xi_1(x)\partial_y, \ldots , \xi_r(x)\partial_y $& P, $(0,\partial_y)$
\\[+0.25ex]
I$_{13}$&$\mathbb{R}\ltimes \mathbb{R}^{r + 1}$ &$   {\partial_y}  ,   y\partial_y,    \xi_1(x)\partial_y, \ldots , \xi_r(x)\partial_y $ & P, $(0,\partial_y)$\\[+0.25ex]
I$_{14}$&$\mathbb{R}\ltimes \mathbb{R}^{r}$ & ${  {\partial_x},   {\eta_1(x)\partial_y}  ,  {\eta_2(x)\partial_y},\ldots ,\eta_r(x)\partial_y} $&P
\\[+0.25ex]
I$_{15}$&$\mathbb{R}^2\ltimes \mathbb{R}^{r}$ &  ${  {\partial_x},    {y\partial_y}  ,    {\eta_1(x)\partial_y},\ldots, \eta_r(x)\partial_y} $&No\\[+0.25ex]
I$_{16}$&$C_\alpha^r\simeq \mathfrak{h}_2\ltimes\mathbb{R}^{r + 1}$ & ${   {\partial_x},    {\partial_y}  ,   x\partial_x  +  \alpha y\partial y,   x\partial_y, \ldots, x^r\partial_y} ,\quad \alpha\in\mathbb{R}$&$(\alpha=-1)$ P
\\[+0.25ex]
I$_{17}$&$\mathbb{R}\ltimes(\mathbb{R}\ltimes \mathbb{R}^{r})$ &$   {\partial_x},    {\partial_y}  ,   x\partial_x  +  (ry  +  x^r)\partial_y ,   x\partial_y, \ldots,  x^{r - 1}\partial_y $ &No
\\[+0.25ex]
I$_{18}$&$(\mathfrak{h}_2\oplus \mathbb{R})\ltimes \mathbb{R}^{r + 1}$ & $   {\partial_x},    {\partial_y} ,   x\partial_x,   x\partial_y,   y\partial_y,   x^2\partial_y, \ldots,x^r\partial_y$ &No
\\[+0.25ex]
I$_{19}$&$\mathfrak{sl}(2 )\ltimes \mathbb{R}^{r + 1}$ &  $   {\partial_x},    {\partial_y}  ,   x\partial_y,    2x\partial _x  +  ry\partial_y,   x^2\partial_x  +  rxy\partial_y,   x^2\partial_y,\ldots, x^r\partial_y  $&No
\\[+0.25ex]
I$_{20}$&$\mathfrak{gl}(2 )\ltimes \mathbb{R}^{r + 1}$ &  $  {\partial_x},    {\partial_y} ,   x\partial_x,   x\partial_y,   y\partial_y,   x^2\partial_x  +  rxy\partial_y,   x^2\partial_y,\ldots, x^r\partial_y $ &No
\\[+1.5ex]
\hline
 \end{tabular}
\hfill}
\end{table}

Let us show that (\ref{coupRiceq}) can be associated with a Jacobi--Lie system for $n=1$. Since every Lie system on $\mathbb{R}$ can be brought into this form through a local diffeomorphism on $\mathbb{R}$, this proves that every Lie system on the real line can be considered as a Jacobi--Lie system. Recall that (\ref{coupRiceq}) is a Lie system with a VG Lie algebra $V$ spanned by (\ref{VF}). Note that $V$ consists of Hamiltonian vector fields
with respect to $(\mathbb{R},
\Lambda=0,R={\partial_{x_1}}).$
Indeed, 
$X_1,X_2,X_3\in V$ admit Hamiltonian functions
$h_1=1, h_2=x_1, h_3=x_1^2.$ Observe that $V$ does not consists of Hamiltonian vector fields relative to any non-zero Poisson bivector on $\mathbb{R}$. Hence, Riccati equations on $\mathbb R$ are not Lie--Hamilton systems but $(\mathbb{R},\Lambda=0,R=\partial_{x_1},a_0(t)X_1+a_1(t)X_2+a_2(t)X_3)$ is a Jacobi--Lie system.

We now classify Jacobi--Lie systems $(\mathbb{R}^2,\Lambda,R,X)$, where we may assume $\Lambda$ and $R$ to be locally equal or different from zero. There exists just one Jacobi--Lie system with $\Lambda=0$ and $R=0$: $(\mathbb{R}^2,\Lambda=0,R=0,X=0)$. Jacobi--Lie systems of the form $(\mathbb{R}^2,\Lambda\neq 0,R=0)$ are Lie--Hamilton systems, whose VG Guldberg Lie algebras were obtained in \cite{LH2013}. In Table \ref{table1} we indicate these cases by writing P ({\it Poisson}). A Jacobi--Lie system $(\mathbb{R}^2,\Lambda=0,R\neq 0,X)$ is such that if $Y\in V^X$, then $Y=fR$ for certain $f\in C^\infty(\mathbb{R}^2)$. All cases of this type can easily be obtained out of the bases given in Table \ref{table1}. We describe them by writing $(0,R)$ at the last column.

Propositions \ref{Prop1} and \ref{Prop2} below show that the VG Lie algebras of Table \ref{table1} that do not fall into the mentioned categories are not VG Lie algebras of Hamiltonian vector fields with respect to  Jacobi manifolds $(\mathbb{R}^2,\Lambda\neq 0,R\neq 0)$. So, every $(\mathbb{R}^2,\Lambda,R,X)$ admits a VG Lie algebra belonging to one of the previously mentioned classes\footnote{To exclude P$_1$ with $\alpha\neq 0$ and I$_{17}$, we need a trivial modification of Proposition \ref{Prop2} using exactly the same line of thought.}.

\begin{lemma} Every Jacobi manifold on the plane with $R\neq 0$ and $\Lambda\neq 0$ admits a local coordinate system $\{s,t\}$ where $R=\partial_{s}$ and $\Lambda=\partial_s\wedge\partial_t$.
\end{lemma}
\begin{proof} Since it is assumed $R\neq 0$, there exist local coordinates $\{s,t_0\}$ on which $R=\partial_s$ and $\Lambda=f(s,t_0)\partial_s\wedge\partial_{t_0}$. Since $[R,\Lambda]_{SN}=0$, we get that $\partial_sf=0$ and $\Lambda=f(t_0)\partial_s\wedge\partial_{t_0}$. As we assume $\Lambda\neq 0$, there exists a new variable $t=t(t_0)$ such that ${\rm d}t/{\rm d}t_0\equiv f^{-1}(t_0)$. Finally, $\Lambda=\partial_{s}\wedge \partial_t$. 
\end{proof}
\begin{definition} We call the local coordinate system $\{s,t\}$ of the above lemma {\it local rectifying coordinates} of the Jacobi manifold on the plane.
\end{definition}
\begin{lemma}\label{63} Let  $(\mathbb{R}^2,\Lambda,R)$ be a Jacobi manifold with $R_\xi\neq 0$ and $\Lambda_\xi\neq 0$ at every $\xi\in \mathbb{R}^2$. The Lie algebra morphism $\phi: C^\infty(\mathbb{R}^2)\ni f\mapsto X_f\in{\rm Ham}(\mathbb{R}^2,\Lambda, R)$ has non-trivial kernel. On local rectifying coordinates, we have $\ker \phi=\langle e^{t}\rangle$.
\end{lemma}
\begin{proof} If $f\in \ker \phi$, then  $\widehat {\Lambda} ({\rm d}f)+fR=0$. In local rectifying coordinates, we get

$
\partial_sf\partial_t-\partial_tf\partial_s+f\partial_s=0\Rightarrow \left\{\begin{array}{c}\partial_sf=0\\\partial_tf=f\end{array}\right.\Rightarrow f=\lambda e^{t},\qquad \lambda\in \mathbb{R}.
$
\end{proof}
\begin{proposition}\label{Prop1} Let $V$ be a VG Lie algebra on $\mathbb{R}^2$ containing $X_1,X_2\in V\backslash\{0\}$ such that
$[X_1,X_2]=X_1$ and $X_1\wedge X_2=0$. Then $V$ does not consist of Hamiltonian vector fields relative to any Jacobi manifold with $R\neq 0$ and $\Lambda\neq 0$.
\end{proposition}
\begin{proof} Assume that $X_1,X_2$ are Hamiltonian. Since $\phi_{\Lambda,R}$ is a Lie algebra morphism, there are $h_1,h_2\in C^\infty(\mathbb{R}^2)\backslash\{0\}$ such that $\{h_1,h_2\}_{\Lambda,R}=h_1+g$, with $g\in \ker \phi$. So,
$$
\{h_1,h_2\}_{\Lambda,R}=\Lambda({\rm d}h_1,{\rm d} h_2)+h_1Rh_2-h_2Rh_1=h_1+g  .
$$
Meanwhile, $X_1\wedge X_2=0$ yields that
$$
\widehat{\Lambda}({\rm d}h_1)\wedge \widehat{\Lambda}({\rm d}h_2)+R\wedge [h_1\widehat {\Lambda}({\rm d}h_2)-h_2\widehat {\Lambda}({\rm d}h_1)]=0.
$$
Using local rectifying coordinates, we see that $\widehat{\Lambda}({\rm d}h_i)=(Rh_i)\partial_t-\partial_th_iR$  and $R\wedge \widehat {\Lambda}({\rm d}h_i)=(Rh_i)\Lambda$ for $i=1,2$. Hence,
$$
\left[(Rh_1)\partial_t-\partial_th_1R\right]\wedge \left[(Rh_2)\partial_t-\partial_th_2R\right] +[h_1(Rh_2)-h_2(Rh_1)]\Lambda =0
$$
and
$$
0=(Rh_1\partial_th_2-Rh_2\partial_th_1)\Lambda +(h_1Rh_2-h_2Rh_1)\Lambda \Leftrightarrow \Lambda({\rm d}h_1,{\rm d}h_2)+h_1Rh_2-h_2Rh_1=0.
$$
This amounts to $\{h_1,h_2\}_{\Lambda,R}=0$, which implies that $0=h_1+g$ and $X_1=0$. This is impossible by assumption and $X_1$ and $X_2$ cannot be Hamiltonian.  
\end{proof}

 \begin{proposition}\label{Prop2} There exists no Jacobi manifold on the plane with $\Lambda\neq 0$ and $R\neq 0$ turning the elements of a Lie algebra diffeomorphic to $V\equiv \langle \partial_x,\partial_y,x\partial_x+\alpha y\partial_y\rangle$ with $\alpha\notin\{0,-1\}$ into Hamiltonian vector fields.
\end{proposition}

\begin{proof} Let us proceed by reduction to absurd. Assume $(N,\Lambda,R)$ to be a Jacobi manifold and $V\subset {\rm Ham}(N,\Lambda,R)$.  Every Lie algebra diffeomorphic to $V$ can be spanned by some vector fields $X_1,X_2, X_3$
diffeomorphic to $\partial_x,\partial_y,x\partial_x+\alpha y \partial_y$. Then, 
$
[X_1,X_2]=0, \ [X_1,X_3]=X_1, \ [X_2,X_3]=\alpha X_2, \ X_1\wedge X_2\neq 0
$
and
\begin{equation}\label{mu}
X_3=\mu_2X_1+\alpha\mu_1X_2,
\end{equation}
with $X_1\mu_1=X_2\mu_2=0$ and ${\rm d}\mu_1\wedge {\rm d}\mu_2\neq 0$. From \eqref{mu} and since $X_1,X_2,X_3$ are Hamiltonian for certain Hamiltonian functions $h_1,h_2,h_3$, correspondingly, we obtain
$$
\widehat{\Lambda}({\rm d}h_3)+h_3R=\mu_2\widehat{\Lambda}({\rm d}h_1)+\mu_2h_1R+\alpha\mu_1\widehat{\Lambda}({\rm d}h_2)+\alpha\mu_1h_2R
$$
and, by means of the rectified expression for $\Lambda$ and $R$, we get
\begin{equation}\label{sys1}
{\partial_s h_3}=\mu_2{\partial_s h_1}+\alpha\mu_1\partial_s h_2,\qquad \partial_t h_3=\mu_2{\partial_t h_1}+\alpha\mu_1{\partial_t h_2}-\mu_2h_1-\alpha\mu_1h_2+h_3.
\end{equation}

Since $[X_1,X_3]=X_1$, then $\{h_1,h_3\}_{\Lambda,R}=h_1+\lambda_1e^t$, where $e^t$ is a function with zero Hamiltonian vector field and $\lambda_1\in\mathbb{R}$. Hence,
$$
h_1+\lambda_1e^t=\{h_1,h_3\}_{\Lambda,R}=\Lambda({\rm d}h_1,{\rm d}h_3)+h_1(Rh_3)-h_3(Rh_1).
$$
Simplifying and using previous expressions (\ref{sys1}), we find that
$$
h_1+\lambda_1e^t=\mu_1\left(\partial_t h_2\partial_s h_1-\partial_t h_1\partial_s h_2+h_1\partial_s h_2-h_2\partial_s h_1\right)=\alpha\mu_1\{h_1,h_2\}_{\Lambda,R}.
$$
As $[X_1,X_2]=0$, then $\{h_1,h_2\}_{\Lambda,R}=\lambda e^t$ for a certain constant $\lambda\in\mathbb{R}$. Hence,
$
h_1=(\alpha\mu_1\lambda-\lambda_1)e^t
$
and $\lambda\neq0 $. Analogously, $\{h_2,h_3\}_{\Lambda,R}=\alpha(h_2+\lambda_2e^t)$ implies 
$$
\alpha(h_2+\lambda_2e^t)=\mu_2\{h_2,h_1\}_{{\Lambda,R}}\Rightarrow h_2=(-\mu_2\lambda/\alpha-\lambda_2)e^t.
$$

Writing the compatibility condition for the system (\ref{sys1}), we reach to
$$
(\alpha+1)\lambda\left(\partial_s \mu_1\partial_t \mu_2-\partial_s \mu_2\partial_t \mu_1\right)=0.
$$
This implies that ${\rm d}\mu_1\wedge {\rm d}\mu_2=0$. Since $X_1\mu_1=X_2\mu_2=0$, we obtain $X_1\wedge X_2=0$, which is impossible by assumption. This finishes the proof.
\end{proof}

Finally, observe that the elements of I$_{12}$ are Hamiltonian relative to $(\mathbb{R}^2,\partial_x\wedge\partial_y,0)$ and $(\mathbb{R}^2,0,\partial_y)$. In the future we aim to study which VG Lie algebras on $\mathbb{R}^2$ admit a similar property. We also plan to devise more applications of our techniques and to analyze other geometric properties of Jacobi--Lie systems. 

\section*{Acknowledgments}
F.J.~Herranz acknowledges partial financial support by the Spanish MINECO under grant  MTM2013-43820-P and   by Junta de Castilla y Le\'on  under   grant BU278U14.   J.~de Lucas acknowledges funding from the National Science Centre (Poland) under grant HARMONIA
DEC-2012/04/M/ST1/00523. C.~Sard\'on acknowledges a fellowship provided by the University of Salamanca.

\end{document}